\documentclass[letterpaper, 10 pt]{ieeeconf}
\IEEEoverridecommandlockouts
\usepackage[english]{babel}
\usepackage{theorem}
\theoremheaderfont{\itshape\bfseries}
{\theorembodyfont{\itshape}
	
	\newtheorem{lemma}{\textbf{Lemma}}
	
	\newtheorem{theorem}{\textbf{Theorem}}
	
	\newtheorem{remark}{\textbf{Remark}}
	
	\newtheorem{problem}{\textbf{Problem}}
}

\usepackage{graphicx}
\usepackage{epsfig}
\usepackage{mathptmx}
\usepackage{times}
\usepackage{amsmath}
\usepackage{amsfonts}
\usepackage{amssymb}
\usepackage{mathrsfs}
\usepackage{xcolor}
\usepackage{graphicx}
\usepackage[ruled,vlined]{algorithm2e}

\newcommand{\T}{^{\mbox{\tiny T}}}
\newcommand{\R}{\mathbb{R}}

\let\leq\leqslant
\let\geq\geqslant

\newenvironment{proof}[1][Proof]%
{\par\noindent\textit{#1:\ }}%
{\hspace*{\fill} \rule{6pt}{6pt}}
\newenvironment{proof*}[1][Proof]%
{\par\noindent\textit{#1:\ }}{}
\newenvironment{Protocol}[1][htb]
{
	\begin{algorithm}[#1]%
	}{\end{algorithm}}

\DeclareMathOperator{\trace}{tr}
\DeclareMathOperator{\re}{Re}

\usepackage{tikz}
\usetikzlibrary{shapes,calc,arrows,patterns,decorations.pathmorphing
	,decorations.markings}
\usetikzlibrary{arrows.meta}

\newenvironment{system}[1]%
{\setlength{\arraycolsep}{0.5mm}\left\{ \; \begin{array}{#1}}%
	{\end{array} \right.}
\newenvironment{system*}[1]%
{\setlength{\arraycolsep}{0.5mm} \begin{array}{#1}}%
	{\end{array}}

\title{\LARGE \textbf{Scale-free $\mathbf{H_2}$ Almost State Synchronization for Homogeneous Networks of Non-Introspective Agents}}

\author{Zhenwei Liu, Ali Saberi, Anton A. Stoorvogel, and Donya Nojavanzadeh%
	\thanks{This work is supported by the Nature Science
		Foundation of Liaoning Province under Grant 2019-MS-116 and the Fundamental Research Funds for
		the Central Universities of China under Grant N2004014.}%
	\thanks{Zhenwei Liu is with College of Information Science and
		Engineering, Northeastern University, Shenyang 110819,
		P. R. China {\tt\small jzlzwsy@gmail.com}} \thanks{Ali Saberi is with
		School of Electrical Engineering and Computer Science, Washington
		State University, Pullman, WA 99164, USA {\tt\small
			saberi@eecs.wsu.edu}} \thanks{Anton A. Stoorvogel is with
		Department of Electrical Engineering, Mathematics and Computer
		Science, University of Twente, P.O. Box 217, Enschede, The
		Netherlands {\tt\small A.A.Stoorvogel@utwente.nl}} \thanks{Donya Nojavanzadeh is with
		School of Electrical Engineering and Computer Science, Washington
		State University, Pullman, WA 99164, USA {\tt\small
			donya.nojavanzadeh@wsu.edu}} }

\begin{document}
	
	\maketitle
	\thispagestyle{empty}
	\pagestyle{empty}

\begin{abstract}
	This paper studies scale-free protocol design for $H_2$ almost state
	synchronization of homogeneous networks of non-introspective agents
	in presence of external disturbances. The necessary and sufficient conditions are provided by designing collaborative linear 
	dynamic protocols. The design is based on localized
	information exchange over the same communication network, which does not need any knowledge of the
	directed network topology and the spectrum of the associated
	Laplacian matrix. Moreover, the proposed protocol is scalable and
	achieves $H_2$ almost synchronization with a given
	arbitrary degree of accuracy for any arbitrary number of agents.
\end{abstract}


\section{Introduction}

In recent decades, the synchronization problem for multi-agent systems
(MAS) has attracted substantial attention due to the wide potential
for applications in several areas such as automotive vehicle control,
satellites/robots formation, sensor networks, and so on. See for
instance the books \cite{ren-book} and \cite{wu-book} or the survey
paper \cite{saber-murray3}.

State synchronization inherently requires homogeneous networks
(i.e. agents which have identical models). Therefore, in this paper we
focus on homogeneous networks. So far, most work has focused on state
synchronization based on diffusive full-state coupling, where the
agent dynamics progress from single- and double-integrator dynamics
(e.g.  \cite{saber-murray2}, \cite{ren}, \cite{ren-beard}) to more
general dynamics (e.g. \cite{scardovi-sepulchre}, \cite{tuna1},
\cite{wieland-kim-allgower}). State synchronization based on diffusive
partial-state coupling has also been considered, including static
design (\cite{liu-saberi-stoorvogel-zhang-ijrnc,
	liu-zhang-saberi-stoorvogel-auto,liu-zhang-saberi-stoorvogel-ejc}),
dynamic design (\cite{kim-shim-back-seo},
\cite{seo-back-kim-shim-iet}, \cite{seo-shim-back},
\cite{stoorvogel-saberi-zhang-auto2017}, \cite{su-huang-tac},
\cite{tuna3}), and the design with additional communication
(\cite{li-soh-xie-tac2019},
\cite{li-duan-chen-huang}, and \cite{scardovi-sepulchre}). Recently,
scale-free collaborative protocol designs are developed for
homogeneous and heterogeneous MAS
\cite{chowdhury-khalil,donya-liu-saberi-stoorvogel-ACC2020} and for
MAS subject to actuator saturation
\cite{liu-saberi-stoorvogel-nojavanzadeh-cdc19}.

Meanwhile, if the agents have absolute measurements of their own
dynamics in addition to relative information from the network, they
are said to be introspective, otherwise, they are called
non-introspective. There exist some results about these two types of
agents, for example, introspective agents
(\cite{kim-shim-seo,yang-saberi-stoorvogel-grip-journal}, etc), and
non-introspective agents
(\cite{grip-yang-saberi-stoorvogel-automatica,
	wieland-sepulchre-allgower}, etc).

Synchronization and almost synchronization in presence of external disturbances are studied in the literature, where three classes of disturbances have been considered namely:
\begin{enumerate}
	 \item Disturbances and measurement noise with known frequencies.
	 \item Deterministic disturbances with finite power.
	 \item Stochastic disturbances with bounded variance.\\
	 \end{enumerate}
 
For disturbances and measurement noise with known frequencies, it is shown in \cite{zhang-saberi-stoorvogel-ACC2015} and \cite{zhang-saberi-stoorvogel-CDC2016} that actually exact synchronization is achievable. This is shown in \cite{zhang-saberi-stoorvogel-ACC2015} for heterogeneous MAS with minimum-phase and non-introspective agents and networks with time-varying directed communication graphs. Then, \cite{zhang-saberi-stoorvogel-CDC2016} extended this results for non-minimum phase agents utilizing localized information exchange. 

For deterministic disturbances with finite power, the notion of $H_\infty$ almost synchronization is introduced by Peymani et.al for homogeneous MAS with non-introspective agents utilizing additional communication exchange \cite{peymani-grip-saberi}. The goal of $H_\infty$ almost synchronization is to reduce the impact of disturbances on the synchronization error to an arbitrarily degree of accuracy (expressed in the $H_\infty$ norm). This work was extended later in \cite{peymani-grip-saberi-wang-fossen,zhang-saberi-grip-stoorvogel,zhang-saberi-stoorvogel-sannuti2} to heterogeneous MAS with non-introspective agents and without the additional communication and for network with time-varying graphs. $H_\infty$ almost synchronization via static protocols is studied in \cite{stoorvogel2018squared} and
\cite{stoorvogel-saberi-liu-nojavanzadeh-ijrnc19} for MAS with passive and passifiable agents. Recently, in \cite{stoorvogel-saberi-zhang-liu-ejc} necessary and sufficient conditions are provided for solvability $H_\infty$ almost synchronization of homogeneous networks with non-introspective agents and without additional communication exchange. Finally, we developed a scale-free framework for $H_\infty$ almost state synchronization for homogeneous network \cite{liu-saberi-stoorvogel-donya-almost-automatica} utilizing suitably designed localized information exchange.

In the case of stochastic disturbances with bounded variance, the concept of stochastic almost synchronization is introduced by \cite{zhang-saberi-stoorvogel-stochastic} and \cite{zhang-saberi-stoorvogel-stochastic2} which in the latter both stochastic disturbance and disturbance with known frequency are present. The idea of stochastic almost synchronization is to reduce the stochastic RMS norm of synchronization error arbitrary small in the presence of colored stochastic disturbances that can be modeled as the output of linear time invariant systems driven by white noise with unit power spectral intensities. By augmenting this model with agent model one can essentially assume that stochastic disturbance is white noise with unit power spectral intensities. In this case under linear protocols the stochastic RMS norm of synchronization error is the $H_2$ norm of the transfer function from disturbance to the synchronization error. As such one can formulate the stochastic almost synchronization equivalently in a deterministic framework requiring to reduce the $H_2$ norm of the transfer function from disturbance to synchronization error arbitrary small. This deterministic approach is referred to as almost $H_2$ synchronization problem which is equivalent to stochastic almost synchronization problem. Recent work on $H_2$ almost synchronization problem are \cite{stoorvogel-saberi-zhang2}, and \cite{stoorvogel-saberi-zhang-liu-ejc} which provided necessary and sufficient conditions for solvability of $H_\infty$ almost synchronization for homogeneous networks with non-introspective agents and without additional communication exchange. Finally, $H_2$ almost synchronization via static protocols is also studied in \cite{stoorvogel-saberi-liu-nojavanzadeh-ijrnc19} for MAS with passive and passifiable agents.

In this paper, we consider stochastic disturbances with bounded variance and we develop scale-free framework to solve $H_2$ almost state
synchronization problem for homogeneous MAS. We design a class of linear parameterized dynamic
protocol utilizing localized information exchange for both networks with full- and
partial-state coupling. The linear dynamic protocol achieves $H_2$ almost state synchronization for any communication network with any
number of agents which contains a
spanning tree. The main contribution of this work is that the protocol design does not require any information of the communication network
such as a lower bound of non-zero eigenvalue of the associated Laplacian
matrix and the number of agents. It is worth to note that, so far in all the works of the literature on $H_2$ almost synchronization, the protocol design requires at least some information about the communication network
such as bounds on the spectrum of associated Laplacian matrix and the number of
agents.

\subsection*{Notations and Background}
Given a matrix $A\in \mathbb{R}^{m\times n}$, $A\T$ and $A^*$ denote
transpose and conjugate transpose of $A$ respectively while $\|A\|_2$
denotes the induced 2-norm (which has submultiplicative property). The $\mathrm{im}(\cdot)$ denote the image of matrix (vector). A square matrix $A$ is said to be Hurwitz stable if all its eigenvalues are in the open left half complex plane. $A\otimes B$ depicts the Kronecker product between $A$ and $B$. $I_n$ denotes the
$n$-dimensional identity matrix and $0_n$ denotes $n\times n$ zero
matrix; sometimes we drop the subscript if the dimension is clear from
the context. For a deterministic continuous-time signal $v(t)$, we denote the $L_2$ norm by $\| v \|_{L_2}$ and its \emph{Root Mean Square (RMS)} value is defined by 
\begin{equation}
\|v\|_{RMS}=\left(\lim_{T\to \infty}\frac{1}{T}\int_{0}^{T}v(t)\T v(t)dt\right)^\frac{1}{2}
\end{equation}   
and for a stochastic signal $v(t)$ which is modeled as wide-sense stationary stochastic process, the $\|v(t)\|_{RMS}$ is given by
\begin{equation}\label{rms-s}
\|v(t)\|_{RMS}=\left( \mathbf{E}[v\T(t)v(t)]\right)^\frac{1}{2}
\end{equation}
where $\mathbf{E}[\cdot]$ stands for the expectation operation. For stochastic signals that approach
wide-sense stationarity as time $t$ goes on to infinity (i.e. for asymptotically wide-sense stationary signals) \eqref{rms-s} is rewritten as 
\begin{equation}\label{rms-s-s}
\|v(t)\|_{RMS}=\left( \lim_{t\to \infty}\mathbf{E}[v\T(t)v(t)]\right)^\frac{1}{2}.
\end{equation}
For a continuous-time system having a $q\times l$ stable transfer function $G(s)$, the $H_2$ norm of $G(s)$ is defined as
\[
\|G\|_{H_2}=\left(\frac{1}{2\pi}\trace\left[\int_{-\infty}^{+\infty}G(j\omega)G^*(j\omega)d\omega\right]^\frac{1}{2}\right)
\]
By Parseval's theorem, $\|G\|_{H_2}$ can be equivalently be defined as
\[
\|G\|_{H_2}=\left(\trace\left[\int_{0}^{+\infty}g(t)g\T(t)dt\right]^\frac{1}{2}\right)
\]
where $g(t)$ is the weighting function or unit impulse (Dirac distribution) response matrix of $G(s)$, as such for single-input single-output system $\|G\|_{H_2}=\|g\|_{L_2}$. The $H_2$ norm of $G(s)$, can be interpreted as the RMS value of the output when the given system is driven by independent zero mean white noise with unit power spectral densities. Note that the $H_2$ norm of a stable transfer function $G(s)$ is finite if and only if it is strictly proper. The $H_\infty$ norm of $G(s)$ is defined as 
\[
\|G\|_{H_\infty}:=\sup_\omega \sigma_{\max}[G(j\omega)]
\]
where $\sigma_{\max}$ is the largest singular value of $G(j\omega)$. Let $\omega(t)$ and $z(t)$ be energy signals which are respectively the input and the corresponding output of the given system. Then, The $H_\infty$ norm of $G(s)$ turns out to coincide with its RMS gain, namely 
\[
\|G\|_{H_\infty}=\|G\|_{RMS\text{ }gain} \sup_{\|\omega\|\neq 0} \frac{\|z\|_{RMS}}{\|\omega\|_{RMS}}
\]

An important property of the $H_\infty$ norm is that it is submultiplicative. That is for transfer functions $G_1$ and $G_2$, we have
\[
\|G_1G_2\|_{H_\infty} \leq\|G_1\|_{H_\infty}\|G_2\|_{H_\infty}.
\]

 A \emph{weighted graph} $\mathcal{G}$ is defined by a triple
$(\mathcal{V}, \mathcal{E}, \mathcal{A})$ where
$\mathcal{V}=\{1,\ldots, N\}$ is a node set, $\mathcal{E}$ is a set of
pairs of nodes indicating connections among nodes, and
$\mathcal{A}=[a_{ij}]\in \mathbb{R}^{N\times N}$ is the weighting
matrix. Each pair in $\mathcal{E}$ is called an \emph{edge}, where
$a_{ij}>0$ denotes an edge $(j,i)\in \mathcal{E}$ from node $j$ to
node $i$ with weight $a_{ij}$. Moreover, $a_{ij}=0$ if there is no
edge from node $j$ to node $i$. We assume there are no self-loops,
i.e.\ we have $a_{ii}=0$. A \emph{path} from node $i_1$ to $i_k$ is a
sequence of nodes $\{i_1,\ldots, i_k\}$ such that
$(i_j, i_{j+1})\in \mathcal{E}$ for $j=1,\ldots, k-1$. A
\emph{directed tree} with root $r$ is a subgraph of the graph
$\mathcal{G}$ in which there exists a unique path from node $r$ to
each node in this subgraph. A \emph{directed spanning tree} is a
directed tree containing all the nodes of the graph. 

For a weighted graph $\mathcal{G}$, the matrix
$L=[\ell_{ij}]$ with
\[
\ell_{ij}=
\begin{system}{cl}
\sum_{k=1}^{N} a_{ik}, & i=j,\\
-a_{ij}, & i\neq j,
\end{system}
\]
is called the \emph{Laplacian matrix} associated with the graph
$\mathcal{G}$. The Laplacian matrix $L$ has all its eigenvalues in the
closed right half plane and at least one eigenvalue at zero associated
with right eigenvector $\textbf{1}$, i.e. a vector with all entries
equal to $1$. When graph contains a spanning tree,  then it follows from \cite[Lemma
$3.3$]{ren-beard} that the Laplacian matrix $L$ has a simple
eigenvalue at the origin, with the corresponding right eigenvector
$\textbf{1}$, and all the other eigenvalues are in the open right-half
complex plane.

\section{Problem formulation}

Consider a MAS composed of $N$ identical linear time-invariant agents
of the form,
\begin{equation}\label{homst-agent-model}
\begin{system*}{cl}
\dot{x}_i &= Ax_i +B u_i+E\omega_i,  \\
y_i &= Cx_i,
\end{system*}\qquad (i=1,\ldots,N)
\end{equation}
where $x_i\in\R^n$, $u_i\in\R^m$, $y_i\in\R^p$ are respectively the
state, input, and output vectors of agent $i$, and
$\omega_i\in\R^{w}$ is the external disturbance.



The communication network provides each agent with a linear
combination of its own outputs relative to that of other neighboring
agents. In particular, each agent $i\in\{1,\ldots,N\}$ has access to the
quantity,
\begin{equation}\label{homst-zeta}
\zeta_i = \sum_{j=1}^{N}a_{ij}(y_i-y_j)=\sum_{j=1}^{N}\ell_{ij}y_j,
\end{equation}
where $a_{ij}\geq 0$ and $a_{ii}=0$ indicate the communication among
agents while $\ell_{ij}$ denote the coefficients of the associated
Laplacian matrix $L$. This communication topology of the network can
be described by a weighted and directed graph $\mathcal{G}$ with nodes
corresponding to the agents in the network and the weight of edges
given by the coefficient $a_{ij}$.  

The MAS \eqref{homst-agent-model} and \eqref{homst-zeta} is referred to as MAS with full-state coupling when $C=I$, otherwise it is called MAS with partial-state coupling.


%
%

Let $N$ be any positive number and define $\bar{x}_i=x_i-x_N$, while
\[
\bar{x}=\begin{pmatrix} \bar{x}_1\\ \vdots \\ \bar{x}_{N-1}
\end{pmatrix}
\quad\text{ and }\quad
\omega=\begin{pmatrix} \omega_1\\ \vdots \\ \omega_N 
\end{pmatrix}.
\]
We denote by $T_{\omega\bar{x}}$ the transfer function from $\omega$
to $\bar{x}$.


In this paper, we introduce an localized
exchange of information among protocols. In particular, each agent 
$i=1,\ldots, N$ has access to localized information, denoted by
$\hat{\zeta}_i$, of the form
\begin{equation}\label{eqa1}
\hat{\zeta}_i=\sum_{j=1}^Na_{ij}(\xi_i-\xi_j)
\end{equation}
where $\xi_j\in\mathbb{R}^n$ is a variable produced internally by
agent $j$ which will be appropriately chosen in the coming sections.\\

In this paper, we focus on scale-free stochastic almost state synchronization for MAS subject to external stochastic disturbances. The scale-free design framework does not require information of the communication topology and the size of the network.

 We adopt a deterministic framework for stochastic almost state synchronization as explained in the Introduction that equivalently we focus on scale-free $H_2$ almost state synchronization problem for MAS subject to external stochastic disturbances. More specifically reducing the stochastic RMS norm of synchronization error to any arbitrary degree of accuracy is equivalent to reducing the $H_2$ norm of the transfer function of synchronization error to the disturbance with desired arbitrary degree of accuracy.

We formulate the scale-free $H_2$ almost state
synchronization problem of a MAS with localized information
exchange. 


\begin{problem}\label{prob4}
  The \textbf{scale-free \boldmath $H_2$ almost state synchronization
    problem with localized information exchange (scale-free
    $H_2$-ASSWLIE)} for MAS \eqref{homst-agent-model} and
  \eqref{homst-zeta} is to find, if possible, a fixed linear protocol
  parameterized in terms of a scalar parameter $\rho$ of the form:
  \begin{equation}\label{protoco1}
  \begin{system}{cl}
  \dot{x}_{i,c}&=A_{c}(\rho)x_{i,c}+B_{c}(\rho){\zeta}_i
  +C_{c}(\rho)\hat{\zeta}_i\\
  u_i&=F_c(\rho)x_{i,c}
  \end{system}
  \end{equation}
  where $\hat{\zeta}_i$ is defined by \eqref{eqa1}, with $\xi_i=H_c x_{i,c}$ with $x_{i,c}\in \mathbb{R}^{n_c}$ such that for any number of agents $N$, and any communication graph $\mathcal{G}$ we have:
  \begin{itemize}
  	\item in the absence of the disturbance $\omega$, for all initial
  	conditions the state synchronization 
  	 \begin{equation}\label{synch_org}
  	\lim_{t\to \infty} (x_i-x_j)=0 \quad  \text{for all $i,j \in {1,...,N}$}.
  	\end{equation}
  	 is achieved
  	for any $\rho\geq1$.
  	\item in the presence of the disturbance $\omega$, for any $\gamma>0$,
  	one can render the $H_2$ norm from $\omega$ to $x_i-x_j$ less
  	than $\gamma$ by choosing $\rho$ sufficiently large.
  \end{itemize}
The architecture of the protocol \eqref{protoco1} is shown in Figure \ref{Architecture}.
\end{problem}

\begin{figure}[t]
	\includegraphics[width=9cm, height=5cm]{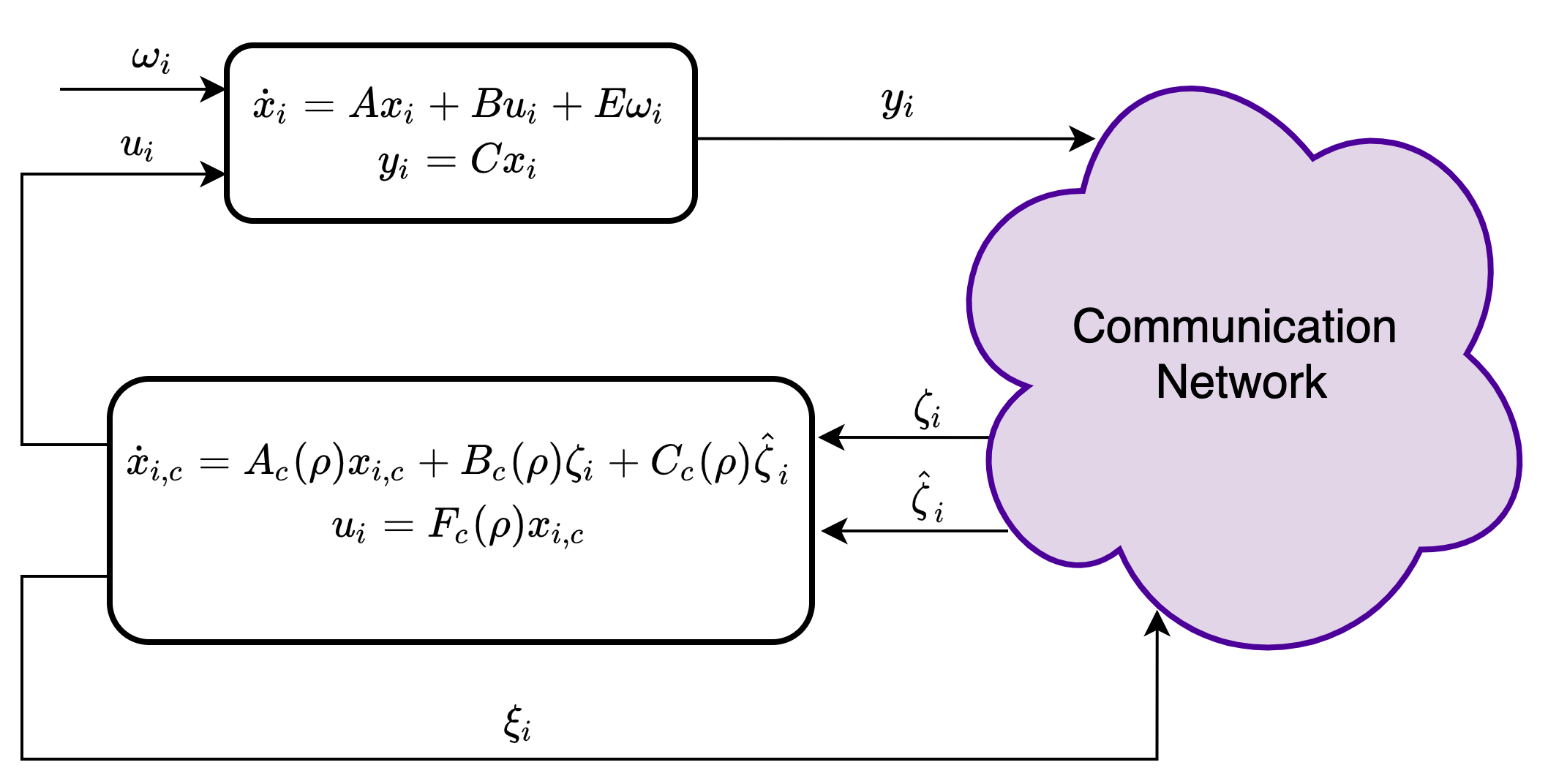}
	\centering
	\caption{Architecture of scale-free collaborative protocol for $H_2$ almost state synchronization}\label{Architecture}
\end{figure}

\begin{remark}
	We would like to emphasize that in our formulation of Problem
	\ref{prob4}, the protocol \eqref{protoco1},
	i.e. $(A_c(\rho),B_c(\rho),C_c(\rho),F_c(\rho))$ to be solely designed based on agent
	model $(A,B,C,E)$ and independent of communication graph and the
	number of agents. 
\end{remark}

\section{\boldmath $H_2$ almost state synchronization: solvability conditions and protocol design}

In this section, we will consider the $H_2$ almost state
synchronization problem of a MAS for both cases of full- and
partial-state coupling.

\subsection{Full-state coupling}
We use the following dynamic protocol with localized information
exchanges.

\begin{Protocol}[h]	
	\caption{Full-state coupling}		
		We design collaborative protocols for agent $i\in\{1,\ldots,N\}$ as
		\begin{equation}\label{pscp1}
		\begin{system}{cll}
		\dot{\chi}_i &=& A\chi_i+Bu_i+\rho{\zeta}_i-\rho\hat{\zeta}_i ,\\
		u_i &=& -\rho B\T P\chi_i
		\end{system}
		\end{equation}
		where $\rho$ is a parameter satisfying $\rho\geq1$ while $P$ is the unique solution of algebraic Riccati equation
		\begin{equation}\label{arespecial}
		A\T P + P A -  PBB\T P +  I = 0 
		\end{equation}
		and ${\zeta}_i$ is defined by \eqref{homst-zeta}. The agents communicate $\xi_i=\chi_i$, therefore each agent has access to local information
		\begin{equation}\label{info1}
		\hat{\zeta}_i=\sum_{j=1}^Na_{ij}(\chi_i-\chi_j).
		\end{equation}
\end{Protocol} 
Then, we have the following theorem for scale-free $H_2$-ASSWLIE case.
\begin{theorem}\label{mainthm2}	
 Consider a MAS described by \eqref{homst-agent-model} and
\eqref{homst-zeta}, where $C=I$.
\begin{enumerate}  \item The scale-free $H_2$-ASSWLIE problem as stated in Problem
	\ref{prob4} is solvable \textbf{if and only if}
	\begin{enumerate}
		\item\label{ass1a} $(A,B)$ is stabilizable.
		\item\label{ass1b} All eigenvalues of $A$ are in the closed left half plane. 
		\item The graph $\mathcal{G}$, describing the communication topology of the network, contains a directed spanning tree.
		\item\begin{equation}\label{eqtc}
		\mathrm{im}E\subseteq\mathrm{im}B.
		\end{equation}
	\end{enumerate}
	\item The linear dynamic Protocol $1$ solves
	scale-free $H_2$-ASSWLIE. In other words, for any number of
	agents $N$ and any graph $\mathcal{G}$  in the absence of the disturbance $\omega$, for any
	$\rho\geq1$, the state synchronization \eqref{synch_org} is achieved
	for any initial conditions while in the presence of the disturbance
	$\omega$, for any $\gamma>0$, the $H_2$ norm from $\omega$ to
	$x_i-x_j$ is less that $\gamma$ by choosing $\rho$ sufficiently
	large.
\end{enumerate}
\end{theorem}

To prove this theorem, we need the following lemma.
\begin{lemma}\label{LbarL}
	Let a Laplacian matrix $L\in \R^{N\times N}$ be given associated
	with a graph that contains a directed spanning tree. We define
	$\bar{L}\in \R^{(N-1)\times (N-1)}$ as the matrix
	$\bar{L}=[\bar{\ell}_{ij}]$ with $\bar{\ell}_{ij} = \ell_{ij}-\ell_{Nj}$. 
	Then the eigenvalues of $\bar{L}$ are equal to the nonzero
	eigenvalues of $L$.
\end{lemma}

\begin{proof}
	We have
	\[
	\bar{L} = \begin{pmatrix} I & -\textbf{1} \end{pmatrix}
	L \begin{pmatrix} I & 0 \end{pmatrix}\T
	\]
	Assume that $\lambda$ is a nonzero eigenvalue of $L$ with
	eigenvector $x$, then
	\[
	\bar{x} = \begin{pmatrix} I & -\textbf{1} \end{pmatrix} x
	\]
	satisfies
	\[
	\begin{pmatrix} I & -\textbf{1} \end{pmatrix} Lx =
	\begin{pmatrix} I & -\textbf{1} \end{pmatrix} \lambda x
	=\lambda \bar{x}
	\]
	for $Lx=\lambda x$ and since $L\textbf{1}=0$ we have $L ( I \quad 0 )\T( I \quad -\textbf{1} )=L$, then we find that 
	\[
	\bar{L} \bar{x} = \begin{pmatrix} I & -\textbf{1} \end{pmatrix}
	L \begin{pmatrix} I & 0 \end{pmatrix}\T\begin{pmatrix} I & -\textbf{1} \end{pmatrix} x=
	\begin{pmatrix} I & -\textbf{1} \end{pmatrix} L x=\lambda \bar{x}.
	\]
	This shows that $\bar{x}$ is an eigenvector of $\bar{L}$ if
	$\lambda\neq0$. It is easily seen that $\bar{x}=0$ if and only if
	$\lambda=0$. Conversely if $\bar{x}$ is an eigenvector of $\bar{L}$
	with eigenvalue $\lambda$ then it is easily verified that $x = L \begin{pmatrix} I & 0 \end{pmatrix}\T \bar{x}$ is an eigenvector of $L$ with eigenvalue $\lambda$.
\end{proof}

\begin{proof*}[Proof of Theorem \ref{mainthm2}]
  Firstly, let $\bar{x}_i=x_i-x_N$ and $\bar{\chi}_i=\chi_i-\chi_N$. We find:
  \[
  \begin{system*}{l}
  \dot{\bar{x}}_i=A{\bar{x}}_i+B(u_i-u_N)+E(\omega_i-\omega_N),\\
  \dot{\bar{\chi}}_i=A\bar{\chi}_i+B(u_i-u_N)
  +\rho\sum_{j=1}^{N-1}\bar{\ell}_{ij}(\bar{x}_j-\bar{\chi}_j),\\ 
  u_i-u_N=-\rho B\T P \bar{\chi}_i.
  \end{system*}
  \]
  Next, we define
  \[
  \bar{x}=\begin{pmatrix}
  \bar{x}_1\\\vdots\\\bar{x}_{N-1}
  \end{pmatrix},\;\;
  \bar{\chi}=\begin{pmatrix}
  \bar{\chi}_1\\\vdots\\\bar{\chi}_{N-1}
  \end{pmatrix} \text{ and } 
  \omega=\begin{pmatrix}
  \omega_1\\\vdots\\\omega_N
  \end{pmatrix}
  \]	
  and we obtain the following closed-loop system
  \begin{equation*}
  \begin{system}{l}
  \dot{\bar{x}}=(I\otimes A) \bar{x}-\rho( I\otimes BB\T
  P)\bar{\chi} +(\Pi\otimes E) \omega \\
  \dot{\bar{\chi}}=(I\otimes A) \bar{\chi}-\rho(I\otimes BB\T
  P)\bar{\chi} +\rho(\bar{L}\otimes I)(\bar{x}-\bar{\chi})
  \end{system}
  \end{equation*}
  where $\bar{L}$ as defined in Lemma \ref{LbarL} and $\Pi=\begin{pmatrix}
  I&-\textbf{1}
  \end{pmatrix}$.
  Let $e=\bar{x}-\bar{\chi}$, we can obtain  
  \begin{align}
  \dot{\bar{x}}&=[I\otimes (A-\rho BB\T P)] \bar{x}  +\rho (I\otimes BB\T P)e +(\Pi\otimes E)
  \omega \label{newsystem2} \\ 
  \dot{e}&=(I\otimes A-\rho\bar{L}\otimes I)e+(\Pi\otimes E)
  \omega \label{newsystem22} 
  \end{align}
  According to Lemma \ref{LbarL}, we have that the real part of the
  eigenvalues of $\bar{L}$ are positive. Therefore, there exists a
  non-singular transformation matrix $T$ such that
  \begin{equation}\label{boundapl}
  (T\otimes I)(I\otimes A-\rho\bar{L}\otimes I)(T^{-1}\otimes I)
  =I\otimes A-\rho\bar{J}\otimes I, 
  \end{equation}
  where
  \[
  \bar{J}=\begin{pmatrix}
  {\lambda}_2&0&\cdots&0\\
  J_{21}&\ddots&\ddots&\vdots\\
  \vdots&\ddots&\ddots&0\\
  J_{N,1}&\cdots&J_{N,N-1}&{\lambda}_N
  \end{pmatrix}
  \]
  with $\re(\lambda_i)>0$ $(i=2,\cdots,N)$, where $\lambda_i$ are the
  nonzero eigenvalues of $L$.
  
  Then, for the stability of \eqref{boundapl}, we just need to prove
  the stability of $ A-\rho\lambda_i I$ for $i=1,\cdots,N-1$. Since the eigenvalues of $A$ are in the closed
  left half plane and $\rho\geq1$, $A-\rho\lambda_i I$ is
  asymptotically stable, i.e. the real part of its eigenvalues are all
  negative. Since \eqref{boundapl} is asymptotically stable, we find
  that $I\otimes A-\rho\bar{L}\otimes I$
  is asymptotically stable. Further, since $A-\rho BB\T P$ is asymptotically stable too, we have the closed-loop system consisting of \eqref{newsystem2} and \eqref{newsystem22} is asymptotically stable without disturbance $\omega$.  
  Namely the protocol
  \eqref{pscp1} can achieve the state synchronization. And we just
  need to prove the $H_2$ norm from $\omega$ to $\bar{x}$ can be made
  arbitrary small.
	
  From \eqref{newsystem2} and \eqref{newsystem22}, we have
  \begin{align*}
    T_{\omega \bar{x}}(s)&=\begin{pmatrix}
      I&0
    \end{pmatrix}\begin{pmatrix}
      T_3&\;\;-\rho I\otimes BB\T P \\0&T_2
    \end{pmatrix}^{-1}\begin{pmatrix}
      \Pi\otimes E\\
      \Pi\otimes E
    \end{pmatrix}\\
    &=T_3^{-1}\left(I+\rho [I\otimes (BB\T P)] T_2^{-1}\right)(\Pi\otimes BX)\\
    &=T_3^{-1}(\Pi\otimes BX)\\
    &\qquad\qquad+\rho T_3^{-1}(I\otimes B) (I\otimes B\T P)
    T_2^{-1}(\Pi\otimes BX)
  \end{align*}
  since we have \eqref{eqtc} such that $E=BX$ for an $X$, where
  \begin{align*}
  T_2&=sI-(I\otimes A-\rho\bar{L}\otimes I)\\
    T_3&=sI-I\otimes (A-\rho BB\T P)
  \end{align*}
  Then, we have
  \begin{multline*}
    \|T_{\omega \bar{x}}(s)\|_{H_2}\leq \|T_3^{-1}[\Pi\otimes B X]\|_{H_2}\\
    +\|\rho T_3^{-1}(I\otimes B) (I\otimes B\T P) T_2^{-1}[\Pi\otimes BX]\|_{H_2}.
  \end{multline*}	
  We know that $T_3^{-1}(\Pi\otimes B X)$ is the transfer
  matrix of the following system: 
  \[
    \dot{p}=[I\otimes (A-\rho BB\T P)]p+(\Pi\otimes B X) w.
  \]
  Thus let 
  \[
    \Theta_0=\frac{\|\Pi\|_2^2\|X\|_2^2}{\rho} P^{-1}
  \]
  with $P$ satisfying \eqref{arespecial}, then we have
  \begin{align*}
    \Theta_0& (A-\rho BB\T P)\T +(A-\rho BB\T P)\Theta_0
    + \Pi\Pi\T \otimes B XX\T B\T\\ 
    &= \frac{\|\Pi\|_2^2\|X\|_2^2}{\rho}(P^{-1}A\T +AP^{-1})-\|\Pi\|_2^2\|X\|_2^2BB\T\\
    &= -\frac{1}{\rho}\|\Pi\|_2^2\|X\|_2^2P^{-2}-(1-\frac{1}{\rho})\|\Pi\|_2^2\|X\|_2^2BB\T<0
  \end{align*}
  since $\rho\geq1$ and $P>0$. According to \cite[Lemma
  3]{stoorvogel-saberi-zhang-liu-ejc}, the $H_2$ norm from $w$ to $p$
  is less than $\rho^{-1} \|\Pi\|_2^2\|X\|_2^2\|P^{-1}\|_2$, i.e.
  \[
    \|T_3^{-1}[\Pi\otimes (B X)]\|_{H_2}< \frac{\|\Pi\|_2^2\|X\|_2^2\|P^{-1}\|_2}{\rho}.
  \]
  Meanwhile, because $T_3^{-1}(I\otimes B)$ is a strictly proper
  transfer matrix, we have
  \begin{multline*}
    \|\rho T_3^{-1}(I\otimes B) (I\otimes B\T P) T_2^{-1}(\Pi\otimes
    BX) \|_{H_2} \\ \leq \|T_3^{-1}(I\otimes B)\|_{H_2}\| \rho (I\otimes B\T
    P) T_2^{-1}(\Pi\otimes BX) \|_{H_\infty}.
  \end{multline*}
  We knows that $T_3^{-1}(I\otimes B)$ is the transfer matrix of the
  following system:
  \[
    \dot{q}=(A-\rho BB\T P)q+B w.
  \]
  Let
    \[
  \Theta=\rho^{-1}P^{-1},
  \]
  with $P$ satisfying \eqref{arespecial}, then we have
  \begin{align*}
    &\Theta (A-\rho BB\T P)\T +(A-\rho BB\T P)\Theta +BB\T\\
    =&{\rho}^{-1}(P^{-1}A\T +AP^{-1})-BB\T\\
    =&-{\rho}^{-1}P^{-2}-(1-{\rho}^{-1})BB\T<0
  \end{align*}
  since $\rho\geq1$ and $P>0$. Correspondingly, the $H_2$ norm from
  $w$ to $q$ is less than $\frac{1}{\rho} \|P^{-1}\|_2$, i.e.
  \[
    \|T_3^{-1}(I\otimes B)\|_{H_2}<\rho^{-1} \|P^{-1}\|_2.
  \]
  Meanwhile, since $\bar{L}$ is invertible,
  there exists a constant $W_2 > 0$ such that
  $\| T_2^{-1}\|_{H_\infty}\leq \frac{W_2}{\rho}$, and thus obtain
  \begin{align*}
    \|T_{\omega \bar{x}}\|_{H_2}\leq&\|T_3^{-1}(\Pi\otimes BX) \|_{H_2}\\
    &\;\;\quad+\|T_3^{-1}(I\otimes B)\|_{H_2}\| \rho(I\otimes B\T P)
    T_2^{-1}(\Pi\otimes BX)\|_{H_\infty}\\
   \leq&{\rho}^{-1}(\|\Pi\|_2^2\|X\|_2^2\|P^{-1}\|_2+W_2\|P^{-1}\|_2\|B\T
    P\|_2\|BX\|_2\|\Pi\|_2) 
  \end{align*}
  i.e.
  \[
    \|T_{\omega \bar{x}}\|_{H_2}< \rho^{-1}\mathbf{M}
  \]
  with 
  \[
    \mathbf{M}=\|P^{-1}\|_2\|\Pi\|_2\left(\|\Pi\|_2\|X\|_2^2+W_2\|B\T P\|_2\|BX\|_2\right)
  \]
  for $\rho\geq1$. It means that we have
  \begin{equation} \tag*{\rule{6pt}{6pt}}
    \|T_{\omega (x_i-x_j)}\|_{H_2}< \rho^{-1}\mathbf{M}.
  \end{equation}
  
  Now, we will prove the necessity. Assume we have a protocol of the
  form \eqref{protoco1} that achieves synchronization for any possible
  graph in the absence of disturbances. It is easily seen that this
  requires that condition (a) is satisfied. On the other hand, we have that
  \[
  \begin{system}{l}
  \dot{\bar{x}}=(I\otimes A) \bar{x}+ (\bar{L}\otimes BF_c(\rho))\bar{\chi}\\
  \dot{\bar{\chi}}=(I\otimes A_c(\rho)) \bar{\chi}+
  (\bar{L}\otimes B_c(\rho)C)\bar{x}+(\bar{L}\otimes C_c(\rho)H_c)\bar{\chi}    
  \end{system}
  \]
  must be asymptotically stable for all possible Laplacian
  matrices. By letting $\bar{L}\rightarrow 0$ we see that in the limit
  the system must have all eigenvalues in the closed left half plane
  which yields that condition (b) must be satisfied. It is well-known
  that state synchronization is impossible to achieve if the network
  does not have a directed spanning tree. Finally, from the result on
  $H_2$ almost disturbance decoupling in \cite[Corollary
  2.4]{saberi-lin-stoorvogel}, we find that (d) is also a necessary
  condition.
\end{proof*}

\subsection{Partial-state coupling}
In this subsection, we will consider $H_2$ almost state
synchronization via partial-state coupling.
\begin{Protocol}[h]	
	\caption{Partial-state coupling}
		We design collaborative protocols for agent $i\in\{1,\ldots,N\}$ as
		\begin{equation}\label{pscp3}
		\begin{system}{cll}
		\dot{\hat{x}}_i &=& A\hat{x}_i-\rho BB\T P\hat{\zeta}_i+\delta^{-2}
		Q_{\rho}C\T({\zeta}_i-C\hat{x}_i) \\ 
		\dot{\chi}_i &=& A\chi_i+Bu_i+\rho \hat{x}_i-\rho \hat{\zeta}_i \\
		u_i &=& -\rho B\T P\chi_i, 
		\end{system}
		\end{equation}
		where $P>0$ is the
		unique solution of \eqref{arespecial}. Since $(A,E,C,0)$ is
		minimum-phase and left invertible, then for any $\rho\geq1$, there
		exists $\delta>0$ small enough such that $Q_{\rho}>0$ is the unique
		solution of
		\begin{equation}\label{arespecial3}
		Q_{\rho}A\T+AQ_{\rho}+EE\T-\delta^{-2}Q_{\rho}C\T CQ_{\rho}+\rho^2Q_{\rho}^2=0.
		\end{equation}
		In this protocol, agents communicate $\xi_i=\chi_i$, i.e. each
		agent has access to localized information \eqref{info1}, while $\zeta_i$ is defined by  \eqref{homst-zeta}.
\end{Protocol} 
\pagebreak

Then, we have the following theorem for MAS via
partial-state coupling.

\begin{theorem}\label{mainthm4}
Consider a MAS described by \eqref{homst-agent-model} and
\eqref{homst-zeta}.
\begin{enumerate}
	\item The scale-free $H_2$-ASSWLIE problem stated in Problem
	\ref{prob4} is solvable \textbf{if and only if}
	\begin{enumerate}
		\item $(A,B)$ are stabilizable and $(C,A)$ are detectable.
		\item All eigenvalues of $A$ are in the closed left half plane. 
		\item 	$(A,E,C,0)$ is minimum phase and left invertible.
		\item The graph $\mathcal{G}$, describing the communication topology of the network, contains a directed spanning tree.
		\item $\mathrm{im}E\subseteq\mathrm{im}B$ (i.e. \eqref{eqtc}).
		\end{enumerate}
	\item  The linear dynamic Protocol $2$ solves Scale-Free
	$H_2$-ASSWLIE, for any number of agents $N$ and any graph
	$\mathcal{G}$ such that in the
	absence of disturbance $\omega$, for any $\rho\geq1$, the state
	synchronization \eqref{synch_org} is achieved for any initial
	conditions and in the presence of disturbance $\omega$, for any
	$\gamma>0$, the $H_2$ norm from $\omega$ to $x_i-x_j$ is less
	that $\gamma$ by choosing $\rho$ sufficiently large.
\end{enumerate}
\end{theorem}

\begin{proof}[Proof of Theorem \ref{mainthm4}]
  Similar to Theorem \ref{mainthm2} and by defining
  $\tilde{x}_i=\hat{x}_i-\hat{x}_N$, we have
  \[
  \begin{system}{cll}
  \dot{\bar{x}}_i&=&A{\bar{x}}_i+B(u_i-u_N)+E(\omega_i-\omega_N)\\
  \dot{\tilde{x}}_i &=&
  A\tilde{x}_i-\rho BB\T P\sum_{j=1}^{N-1}\bar{\ell}_{ij}\bar{\chi}_j+\delta^{-2}Q_{\rho}C\T
  C(\sum_{j=1}^{N-1}\bar{\ell}_{ij}\bar{x}_j-\tilde{x}_i) \\ 
  \dot{\bar{\chi}}_i &=&
  A\bar{\chi}_i+B(u_i-u_N)+\rho\tilde{x}_i-\rho\sum_{j=1}^{N-1}\bar{\ell}_{ij}\bar{\chi}_j  
  \end{system}
  \]
  We define
  \[
  \bar{x}=\begin{pmatrix}
  \bar{x}_1\\\vdots\\\bar{x}_{N-1}
  \end{pmatrix},\;\;
  \tilde{x}=\begin{pmatrix}
  \tilde{x}_1\\\vdots\\\tilde{x}_{N-1}
  \end{pmatrix},\;\; 
  \bar{\chi}=\begin{pmatrix}
  \bar{\chi}_1\\\vdots\\\bar{\chi}_{N-1}
  \end{pmatrix}, 
  \omega=\begin{pmatrix}
  \omega_1\\\vdots\\\omega_N
  \end{pmatrix}
  \]
  then we have the following closed-loop system
  \[
  \begin{system*}{l}
  \dot{\bar{x}}=(I\otimes A) \bar{x}-\rho[I\otimes BB\T
  P]\bar{\chi}+(\Pi \otimes E)\omega\\ 
  \dot{\tilde{x}} = [I\otimes (A-\frac{1}{\delta^2}Q_{\rho}C\T
  C)]\tilde{x}-\rho(\bar{L}\otimes BB\T P)\bar{\chi} \\
  \hspace*{5cm}+\frac{1}{\delta^2}(\bar{L}\otimes Q_{\rho}C\T C)\bar{x} \\
  \dot{\bar{\chi}} = (I\otimes A-\rho\bar{L}\otimes
  I)\bar{\chi}-\rho(I\otimes BB\T P) \bar{\chi}+\rho\tilde{x} 
  \end{system*}
  \]	
  By defining $e=\bar{x}-\bar{\chi}$ and
  $\bar{e}=(\bar{L}\otimes I)\bar{x}-\tilde{x}$, we can obtain
  \begin{equation}\label{newsystem3}
  \begin{system*}{l}
  \dot{\bar{x}}=[I\otimes (A-\rho BB\T P)] \bar{x}+\rho
  (I\otimes BB\T P)e+(\Pi \otimes E)\omega\\ 
  \dot{\bar{e}}=[I\otimes (A-\delta^{-2}Q_{\rho}C\T
  C)]\bar{e}+(\bar{L} \Pi \otimes E)\omega\\ 
  \dot{e}=(I\otimes A-\rho\bar{L}\otimes I)e+\rho\bar{e}+(\Pi \otimes E)\omega
  \end{system*}
  \end{equation}	
  From Theorem \ref{mainthm2} and \eqref{arespecial3}, we have $A-\rho BB\T P$, $I\otimes A-\rho\bar{L}\otimes I$ and $A-\delta^{-2}Q_{\rho}C\T
  C$ are stable for $\rho\geq 1$ and some a $\delta>0$. It means that system \eqref{newsystem3} is asymptotically stable without disturbance $\omega$. So,   
  we have the protocol
  \eqref{pscp3} can achieve the state synchronization. Meanwhile, we can obtain the following result by choosing the Lyapunov function  
  \[
  V_0=\bar{e}\T (I\otimes Q_\rho^{-1})\bar{e}
  \]
  for $\bar{e}$ with $Q_\rho$ satisfying \eqref{arespecial3}. Then we have
  \begin{align*}
  \dot{V}_0=&\bar{e}\T [I\otimes (Q_\rho^{-1}A+A\T
  Q_\rho^{-1}-2\delta^{-2}C\T C)]\bar{e}\\ 
  &\qquad+2\bar{e}(\bar{L}\Pi \otimes Q_\rho^{-1}E) \omega\\
  \leq&-\bar{e}\T [I\otimes (\rho^2 I+Q_\rho^{-1}EE\T Q_\rho^{-1})]\bar{e}\\
  &\qquad+\bar{e}\T (I\otimes Q_\rho^{-1}EE\T
  Q_\rho^{-1})\bar{e}+\omega\T (\Pi\T\bar{L}\T\bar{L}\Pi \otimes
  I)\omega\\ 
  \leq& -\rho^2 \|\bar{e}\|_{L_2}^2+\|\bar{L}\|_2^2\|\Pi\|_2^2\|\omega\|_{L_2}^2
  \end{align*}
  with $\rho\geq1$. By integrating the above inequality, we have
  \[
  \int_0^\infty -\rho^2 \| \bar{e}(t) \|_{L_2}^2 +
  \|\bar{L}\|_2^2\|\Pi\|_2^2\| \omega(t) \|_{L_2}^2 \text{d}t \geq 0
  \]
  for zero initial conditions and hence
  \[
  \rho^2 \| \bar{e} \|_{L_2}^2 \leq \|\bar{L}\|^2\|\Pi\|^2\| \omega  \|_{L_2}^2
  \]
  i.e. 
  \begin{equation}\label{tweinf}
  \| T_{\omega
  	\bar{e}}\|_{H_\infty}=\frac{\|\bar{e}\|_{L_2}}{\|\omega\|_{L_2}}\leq
  \frac{\|\bar{L}\|_2\|\Pi\|_2}{\rho}	 
  \end{equation}
  for $\rho\geq1$.

  Then we need to
  prove the $H_2$ norm from $\omega$ to $\bar{x}$ (or $x_i-x_j$) can
  be made arbitrary small.	
  From \eqref{newsystem3}, we have
  \begin{align*}
    T_{\omega \bar{x}}(s)&=\begin{pmatrix}
      I&0&0
    \end{pmatrix}\begin{pmatrix}
      T_3&0&-\rho T_4\\
      0&T_1&0\\
      0&-\rho I&T_2
    \end{pmatrix}^{-1}\begin{pmatrix}
      \Pi\otimes E \\
      \bar{L}\Pi\otimes E \\
      \Pi\otimes E
    \end{pmatrix}\\
    &=T_3^{-1}(I\otimes B)T_0
  \end{align*}
  where $T_1=sI-I\otimes (A-\delta^{-2}Q_{\rho}C\T C)$, $T_2,T_3$ are as same as the definition in Theorem \ref{mainthm2}, $T_4=I\otimes BB\T P$, 
  \begin{multline*}
    T_0=(\Pi\otimes X)+\rho^2 (I\otimes B\T P)T_2^{-1} T_{\omega
      \bar{e}}\\+\rho (I\otimes B\T P)T_2^{-1}(\Pi\otimes BX),  
  \end{multline*}
  and $T_{\omega \bar{e}}=T_1^{-1}[(\bar{L}\Pi)\otimes E]$.
	
  Similar to the proof of Theorem \ref{mainthm2}, we have
  \begin{align*}
    \|T_{\omega \bar{x}}\|_{H_2}\leq& \|T_3^{-1}[\Pi\otimes (BX)]\|_{H_2}\\
    &+\|T_3^{-1}(I\otimes B)\rho^2 (I\otimes B\T P)T_2^{-1}
    T_{\omega \bar{e}}\|_{H_2}\\ 
    &+\|T_3^{-1}(I\otimes B)\rho [I\otimes (B\T
    P)]T_2^{-1}(\Pi\otimes BX)\|_{H_2} \\
    \leq& \|T_3^{-1}[\Pi\otimes (BX)]\|_{H_2}\\
    &+\|T_3^{-1}(I\otimes B)\|_{H_2}\|\rho^2 (I\otimes B\T P)T_2^{-1}
    T_{\omega \bar{e}}\|_{H_\infty}\\ 
    &+\|T_3^{-1}(I\otimes B)\|_{H_2}\|\rho [I\otimes (B\T
    P)]T_2^{-1}(\Pi\otimes BX)\|_{H_\infty}
  \end{align*}
  because $T_3^{-1}(I\otimes B)$ is a strictly proper transfer matrix.
  Then, we have
  \begin{align*}
    \|T_{\omega \bar{x}}\|_{H_2}\leq& \|T_3^{-1}[\Pi\otimes (BX)]\|_{H_2}\\
    &+\rho^2 \|T_3^{-1}(I\otimes B)\|_{H_2}\| B\T P\|_2\|T_2^{-1}\|_{H_\infty}\|
    T_{\omega \bar{e}}\|_{H_\infty}\\ 
    &+\rho\|T_3^{-1}(I\otimes B)\|_{H_2}\| B\T P\|_2\|T_2^{-1}\|_{H_\infty}\|\Pi\|_2\|BX\|_2.
  \end{align*}
  From Theorem \ref{mainthm2}, we have 
  \begin{align*}
    \|T_3^{-1}[\Pi\otimes BX]
    \|_{H_2} &<\rho^{-1}\|\Pi\|_2^2\|X\|_2^2\|P^{-1}\|_2\quad  \text{ and }\\
    \|T_3^{-1}(I\otimes B)\|_{H_2} &<\rho^{-1}\|P^{-1}\|_2
  \end{align*}
  for $\rho\geq1$. Meanwhile, according to the results of Theorem
  \ref{mainthm2}, we have
  \[
    \|T_2^{-1}\|_{H_\infty}\leq \rho^{-1} W_2
  \]
  for $\rho\geq1$. Thus, for $\rho\geq1$, we can obtain
  \[
    \|T_{\omega (x_i-x_j)}\|_{H_2}<\rho^{-1}\mathbf{W}
  \]
  with 
  \begin{multline*}
    \mathbf{W}=\|P^{-1}\|_2\|\Pi\|_2\left(\|\Pi\|_2\|X\|_2^2\right.\\\left.
      +W_2\| B\T P\|_2\|\bar{L}\|_2+W_2\| B\T P\|_2\|BX\|_2\right).
  \end{multline*}
  for $\rho\geq1$, i.e.
  \begin{equation} \tag*{\rule{6pt}{6pt}}
    \|T_{\omega (x_i-x_j)}\|_{H_2}< \rho^{-1}\mathbf{W}.
  \end{equation}
  
  As the next step, we will prove the necessity. Similar to the proof
  of Theorem \ref{mainthm2}, we have that (a), (b) and (d) are
  necessary conditions. From the result on $H_2$ almost
  disturbance decoupling in \cite[Theorem 2.3]{saberi-lin-stoorvogel},
  we find that (c) and (e) are also necessary conditions in case of
  partial-state coupling.
\end{proof}

\begin{remark}
	We would like to emphasize that Protocol $1$ and $2$ do not need any information about
	the number of agents and Laplacian matrix $L$ associated with the
	communication graph. On the other hand, the parameter $\rho$ in our
	protocol design is used to reduce the impact of disturbance on synchronization error.  However, when we need a better disturbance
	rejection level (smaller $H_2$ norm), then we need
	to increase $\rho$. It worth to mention that in the absence of disturbance we can choose $\rho$
	arbitrarily.
\end{remark}

\begin{remark}
	Note that Protocol $1$ and $2$ can generate
	a signal which is part of their state and communicate with their neighbors over the same communication graph that leads to achieving
	scalable protocols. Meanwhile, these protocols are universal since
	they can work for any communication graph with any number of agents
	as long as the graph contains a directed spanning tree. Meanwhile
	the design methodology is scalable since it is one shot design based
	on an explicit linear structure.
\end{remark}


\begin{figure}[t]
	\includegraphics[width=9cm, height=10cm]{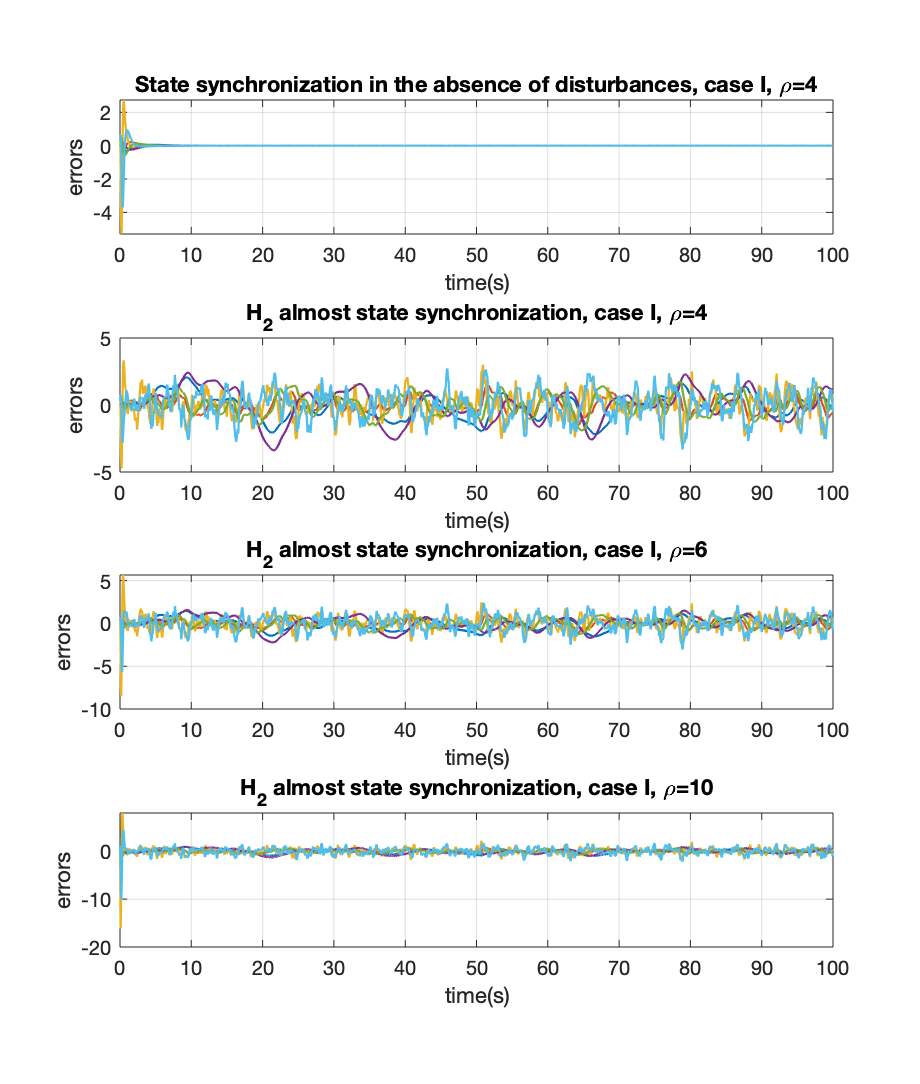}
	\centering
	\vspace*{-10mm}
	\caption{Results of state synchronization and $H_2$ almost state synchronization for the MAS with $N=3$}\label{CaseI}
\end{figure}
\begin{figure}[t]
	\includegraphics[width=9cm, height=10cm]{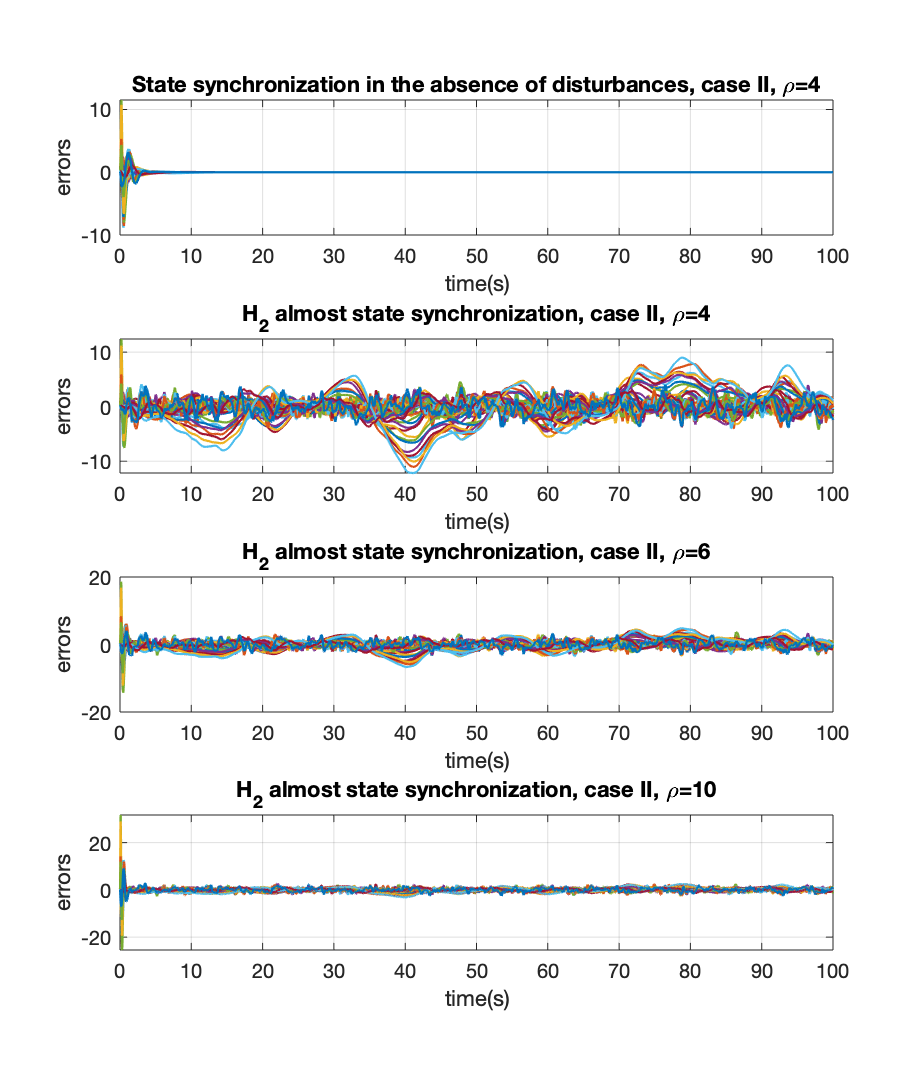}
	\centering
	\vspace*{-10mm}
	\caption{Results of state synchronization and $H_2$ almost state synchronization for the MAS with $N=20$}\label{CaseII}
\end{figure}

\section{Numerical example}

In this section we will illustrate the effectiveness of our protocol
design with a numerical example for $H_2$ state synchronization
with partial-state coupling. Consider agent models \eqref{homst-agent-model} with
\[
A=\begin{pmatrix}
0&1&0\\0&0&1\\0&0&0
\end{pmatrix}, \quad B=\begin{pmatrix}
0\\0\\1
\end{pmatrix}, \quad
C=\begin{pmatrix}
1&0&0
\end{pmatrix}, \quad E=B.
\]
For this agent model, we obtain Protocol $2$, by solving algebraic Riccati equations \eqref{arespecial} and \eqref{arespecial3} for three
values of $\rho=4$, $\rho=6$ and $\rho=10$ (The Riccati equation \eqref{arespecial3} is solved with $\delta= 0.0004$). 

We create two homogeneous MAS with different number of
agents and different communication topologies to show the
designed protocol is scale-free, i.e. it is independent of information of the  communication network and the number of agents $N$.
\begin{itemize}
	\item \textit{Case I}: In this case, we consider MAS with $3$ agents and communication topology $\mathcal{A}_1$, with $a_{21}=a_{32}=1$. The result of exact state synchronization in the absence of disturbance and $H_2$ almost state synchronization in presence of white noises with unit power spectral densities for $i=1,\cdots, N$ are shown in Figure \ref{CaseI}. The results show that by increasing $\rho$, one can decrease the impact of disturbances on synchronization error.
	
	\item \textit{Case II}: Next, we consider a MAS with $20$ agents and associated adjacency matrix $\mathcal{A}_2$, with $a_{16}=a_{21}=a_{32}=a_{43}=a_{54}=a_{65}=a_{76}=a_{87}=a_{98}=a_{10,9}=a_{11,10}=a_{12,11}=a_{13,12}=a_{13,20}=a_{14,13}=a_{15,14}=a_{15,6}=a_{16,15}=a_{17,16}=a_{18,17}=a_{19,18}=a_{20,18}=1$. Figure \ref{CaseII} shows the results for exact state synchronization in the
	absence of disturbance with $\rho=4$ and $H_2$ almost state
	synchronization results in presence of disturbances $\omega_i$ equal to white noises with unit power spectral densities for $i=1,\cdots, N$ with $\rho=4$, $\rho=6$, and $\rho=10$.
\end{itemize}
The simulation results show that the protocol design is independent of
the communication graph and is scale free so that we can achieve
$H_2$ almost state synchronization with one-shot protocol design, for any graph with any number of
agents. The simulation results also show that by increasing the value
of $\rho$, almost state synchronization is achieved with higher degree of
accuracy.

\section{Conclusion}
In this paper, we studied $H_2$ almost state synchronization
of homogeneous networks of non-introspective agents. A parameterized scale-free linear dynamic protocol, parameterized in scalar $\rho$, was developed using localized information exchange over the same communication network and solely based on agent models. In particular, in the absence of disturbance, we achieved synchronization for any $\rho>1$ and in the presence of disturbance we achieved almost state synchronization for a given arbitrary degree of accuracy by choosing $\rho$ sufficiently large. Despite all the existing results, our design methodology was scale-free so that we did not need any information about the communication network such as bounds on the associated Laplacian matrix and the number of agents. As our future work, we aim to extend the scale-free designs proposed in this paper to the broader classes of agent models.

\bibliographystyle{plain}
\bibliography{referenc}
\end{document}